\documentclass[12pt,a4paper]{article}
\usepackage[english]{babel}
\usepackage[T1]{fontenc}
\usepackage{graphicx}
\usepackage{amsmath,amsthm,amssymb,mathrsfs}
\usepackage{epstopdf}
\usepackage{caption}
\usepackage{subcaption}
\usepackage{float}
\usepackage{appendix}
\usepackage{pdfpages}
\usepackage{tikz}
\usepackage{braket}
\usepackage{anyfontsize}
\usepackage{mathptmx}
\usepackage[11pt]{moresize}
\usepackage{enumerate}
\usepackage{dcolumn}
\usepackage{bm}
\usepackage{tabularx}
\usepackage{array}
\usepackage{lipsum}
\usepackage{cancel}
\usepackage{enumitem}
\usepackage{soul}
\usepackage{authblk}
\usepackage{pgfplots}
\usepackage{cite}
\usepackage{ulem}

\def\cH{{\mathcal H}}
\def\cN{{\mathcal N}}

\newcommand{\be}{\begin{equation}}
\newcommand{\ee}{\end{equation}}
\newcommand{\ba}{\begin{align}}
\newcommand{\ea}{\end{align}}
\newcommand{\bea}{\begin{eqnarray}}
\newcommand{\eea}{\end{eqnarray}}
\newcommand{\bt}{\begin{theorem}}
\newcommand{\et}{\end{theorem}}
\newcommand{\bp}{\begin{proposition}}
\newcommand{\ep}{\end{proposition}}
\newcommand{\bd}{\begin{definition}}
\newcommand{\ed}{\end{definition}}

\newtheorem{theorem}{Theorem}
\newtheorem{proposition}{Proposition}
\newtheorem{definition}{Definition}

\newtheorem{fact}{Fact}
\newtheorem{observation}{Observation}

\begin{document}

\title{Discrimination of dephasing channels}

\author[1]{Milajiguli Rexiti}

\affil[1]{\it\small School of Science and Technology, University of Camerino, I-62032 Camerino, Italy}

\author[2]{Laleh Memarzadeh}

\affil[2]{\it\small Department of Physics, Sharif University of Technology, Tehran, Iran}

\author[1,3]{Stefano Mancini}

\affil[3]{\it\small INFN Sezione di Perugia, I-06123 Perugia, Italy}

\date{\today}  

\maketitle

\begin{abstract}
The problem of dephasing channel discrimination is addressed for finite-dimensional systems.
In particular, the optimization with respect to input states without energy constraint is solved analytically for qubit, qutrit and ququart. Additionally, it is shown that resorting to side entanglement assisted strategy is completely useless in this case.
\end{abstract}


\section{Introduction}

Since any physical process can be described as a quantum channel (linear, completely positive and trace preserving map on the set of trace class operators), the issue of quantum channels discrimination started to become pervasive in various fields{ \cite{GLN,MFS,WY,Hay,Pirandola}}. 
It involves a double optimization: on the measurements to be performed at the output and on the input states. Such a task becomes challenging especially when dealing with infinite dimensional systems, namely with continuous variable channels.

Sometimes, finite channels are investigated to get an approximate behaviour of continuous channels. A notable example is the amplitude-damping channel employed as approximation of a lossy channel in the discrimination problem {\cite{MilaJPA,ALS,Fanizza,QZH}}.
The usage of finite channels to get insights about continuous channels would be even more useful when the latter are non-Gaussian. One of the first examples of non-Gaussian continuous variable channel studied in quantum information theory is the dephasing channel \cite{AMM}. It describes the wash out of coherence properties (off diagonal terms) of a state with respect to the Fock basis. 

Here we address the problem of dephasing channels discrimination for finite dimensional systems. As figure of merit the trace distance between the output states corresponding to two possible channels will be used. In other words, Helstrom measurement strategy will be adopted \cite{Hel}. Then, the optimization with respect to input probe will be solved for qubit, qutrit and ququart, even in presence of energy constraint. Additionally, we show that resorting to entanglement-assisted strategy is completely useless in this case.


\section{The model}

Consider a complex Hilbert space $\cH$ of dimension $d$.
Let ${\cal B}=\{ |0\rangle,|1\rangle, \ldots, |d-1\rangle\}$ denote the Fock basis,
namely the basis composed by eigenvectors of the energy operator (Hamiltonian) $H$:
\begin{equation}
\label{eq:H}
H|n\rangle=n|n\rangle, \quad n=0,\ldots, d-1.
\end{equation}
The dephasing channel we are going to consider is the finite dimensional version 
of the bosonic channel discussed in \cite{AMM}. { A more general version of dephasing channel can be found in\cite{Amosov1,Amosov2}}.
Here, it is a completely positive and trace preserving map defined on the set of density operators over $\cH$ as
\be
\rho\rightarrow \cN_\gamma \left(\rho\right)=\sum_{j=0} ^{\infty}K_j\rho K_j ^{\dag},
\ee
with Kraus operators
\be
K_j=e^{-\frac{1}{2}\gamma H^2}\frac{\left(-i \sqrt{\gamma} H \right)^j}{\sqrt{j !}}.
\ee
Here $\gamma\in[0,+\infty)$ is the dephasing parameter.

After writing the input state in the ${\cal B}$ basis as
\be\label{rhoFock}
\rho=\sum_{m,n=0}^{d-1} \rho_{m,n}\ket m \bra n,
\ee
the channel action reads
\be\label{cn}
\cN_\gamma \left(\rho\right)=\sum_{m,n=0}^{d-1} e^{-\frac{1}{2}\gamma \left(m-n\right)^2} \rho_{m,n}\ket m \bra n.
\ee
Now, assume to have two dephasing channels characterized by parameters $\gamma_0$ and $\gamma_1$,
each appearing with probability $\frac{1}{2}$, and we want to discriminate between them. 

According to \cite{Hel}, given two output states $\cN_{\gamma_0} \left(\rho\right)$ and  
$\cN_{\gamma_1} \left(\rho\right)$ with equal probability, the optimal probability of success in discriminating them reads
\be\label{fm}
P_{s}{\left(\rho\right)}=\frac{1}{2}\left(1+\frac{1}{2}\left\| \cN_{\gamma_0} \left(\rho\right)- \cN_{\gamma_1} \left(\rho\right) \right\|_1\right),
\ee
where
\be
\| T \|_1\equiv {\rm Tr}\sqrt{T^{\dag}T}.
\ee
So we are left with the problem of finding 
\be\label{Psf}
\overline{P_s}:=\max_\rho P_s(\rho).
\ee
Below we shall also consider this optimization problem constrained by a fixed 
amount of input average energy
\be\label{Econstr}
E={\rm Tr}\left(H \rho\right). 
\ee
Clearly, it must be $0\leq E\leq d-1$.


\section{Preliminaries}\label{sec:pre}

Before looking for the optimal input states in specific low dimensional Hilbert spaces, 
we would like to make some general observations.

\begin{observation}
To optimize the probability of success over the set of input states to find 
$\overline{P_s}$ in \eqref{Psf} , 
it is enough to consider pure states only. 
\end{observation}

In fact, 
given the spectral decomposition $\rho=\sum_i p_i \varphi_i$, with $\varphi_i= \ket\varphi_i \bra\varphi_i$ ($ {\rm rank}{\varphi_i}=1$), 
thanks to the linearity of channel maps and
the convexity of trace norm,
we have:
\begin{align}\label{eq:convex}
\left\| \cN_{\gamma_0} \left(\sum_i p_i\varphi_i\right)- \cN_{\gamma_1} \left(\sum_i p_i\varphi_i\right) \right\|_1
&=\left\| \sum_i p_i\cN_{\gamma_0} \left(\varphi_i\right)- \sum_i p_i\cN_{\gamma_1} \left(\varphi_i\right) \right\|_1 \notag\\
&\leq\sum_i p_i\left\| \cN_{\gamma_0} \left(\varphi_i\right)- \cN_{\gamma_1} \left(\varphi_i\right) \right\|_1 \notag\\
&\leq \left(\sum_i p_i\right)
\max_{\varphi_i}
\left\| \cN_{\gamma_0} \left(\varphi_i\right)- \cN_{\gamma_1} \left(\varphi_i\right) \right\|_1\notag\\
&=\left\| \cN_{\gamma_0} \left(\varphi^*\right)- \cN_{\gamma_1} \left(\varphi^*\right) \right\|_1,
\end{align}
where we denoted $\varphi^* \equiv  {\rm argmax}_{\varphi_i}
\left\| \cN_{\gamma_0} \left(\varphi_i\right)- \cN_{\gamma_1} \left(\varphi_i\right) \right\|_1$.
As consequence of \eqref{eq:convex} we obtain
\begin{align}
\max_\rho \left\| \cN_{\gamma_0} \left(\rho\right)- \cN_{\gamma_1} \left(\rho\right) \right\|_1\leq \max_{\varphi \,:\, {\rm rank}{\varphi}=1} \left\| \cN_{\gamma_0} \left(\varphi\right)- \cN_{\gamma_1} \left(\varphi\right) \right\|_1.
\end{align}

While purity of input states suffices for generic channel discrimination, below we show that for discriminating dephasing channels, further restrictions on the set of input states can be made.

\begin{fact}\label{property1} An arbitrary input pure state $\ket{\varphi}$ 
written in the basis ${\cal B}$ as
\begin{equation}\label{phi}
\ket{\varphi}=\sum_{j=0}^{d-1}\sqrt{r_j} e^{i\theta_j} |j\rangle,
\end{equation}
with $r_j\in\mathbb{R}^+_0$ and $\theta_j\in[0,2\pi)$, 
performs the same for discriminating dephasing channels as
the input state $U_\varphi\ket{\varphi}$, where
\begin{equation}\label{Uphi}
U_{{\varphi}}\equiv\sum_{k=0}^{d-1} e^{-i\theta_k} \ket{k}\bra{k}.
\end{equation}
\end{fact}
\begin{proof}
By using the representation of dephasing channel as in Eq.~(\ref{cn}), it can be easily seen that
\begin{equation}
\label{eq:semiCov}
{\cal N}_{\gamma}\left(U_{\varphi} {\varphi} U_{{\varphi}}^\dag \right)
=U_{{\varphi}}{\cal N}_{\gamma}({\varphi} )U_{{\varphi}}^\dag,
\end{equation}
{for all possible choices of phases in Eq.(\ref{phi}) and (\ref{Uphi}), given that the basis ${\cal B}$ is fixed.}
Then, taking into account the norm invariance under unitary conjugation, we have
\begin{equation}
\label{eq:equiNorm}
 \left\| {\cal N}_{\gamma_0}\left(U_{\varphi} {\varphi} U_{{\varphi}}^\dag \right)
 - {\cal N}_{\gamma_1}\left(U_{\varphi} {\varphi} U_{{\varphi}}^\dag \right)\right\|_1
 = \left\| \cN_{\gamma_0} ({\varphi})- \cN_{\gamma_1}({\varphi}) \right\|_1.
\end{equation}
From Eq.~(\ref{fm}) and Eq.~(\ref{eq:equiNorm}) we conclude that the states $\ket{\varphi}$ and $U_\varphi\ket{\varphi}$ perform the
same for channel discrimination. 
\end{proof}

\begin{observation}
As a consequence of Fact \ref{property1}, we can restrict the search for optimal input states, 
when expanded in the basis $\mathcal{B}$,
over pure states with non-negative real coefficients subjected to normalization constraint:
\begin{equation}
\label{eq:NC}
\ket{\varphi}=\sum_{n=0}^{d-1}\sqrt{r_n}  \ket{n},\qquad r_n\geq 0, \quad \text{s.t.} \quad \sum_{n=0}^{d-1}r_n=1.
\end{equation}
\end{observation}

Yet another simplification is possible in the optimization in Eq.~\eqref{Psf}. 
In fact, having $\ket\varphi$ as in Eq.~\eqref{eq:NC}
we can see the action of the channel 
\be\label{Ngphi}
{\cal N}_\gamma(\varphi)=\sum_{n,m=0}^{d-1}\sqrt{r_nr_m}e^{-(n-m)^2\gamma/2}\ket{n}\bra{m},
\ee
is the same on each $k$-diagonal of the density matrix $\varphi$.\footnote{Notice that $k=m-n$, but it is enough to consider $k\geq 0$, because the matrix \eqref{Ngphi} is symmetric.
} Thus, to optimize the input, there is no reason to take different 
coefficients within the same $k$-diagonal.

Imposing that the elements of each $k$-diagonal of $\varphi$ are equal amounts to have
\begin{align}\label{eq:sys}
\left\{
\begin{array}{l}
r_0r_1=r_1r_2=r_2r_3=\ldots \qquad (k=1) \\
r_0r_2=r_1r_3=r_2r_4=\ldots \qquad (k=2) \\
r_0r_3=r_1r_4=r_2r_5=\ldots \qquad (k=3) \\
\vdots
\end{array}
\right.
\end{align}
This system of equations \eqref{eq:sys} admits a unique solution
\be\label{eq:coefsym}
r_0=r_{d-1}, \quad r_1=r_{d-2}, \quad r_2=r_{d-3}, \quad \ldots
\ee

\begin{observation}\label{ob3}
As consequence of \eqref{eq:coefsym}, we can restrict the optimization \eqref{Psf} to input states of the form 
\be
\ket{\varphi}=\sum_{j=0}^{\lfloor \frac{d-1}{2}\rfloor} \sqrt{r_j} \ket{j}
+\sum_{j=\lfloor \frac{d-1}{2}\rfloor+1}^{d-1} \sqrt{r_{d-1-j}} \ket{j},
\ee
characterized by $\lceil \frac{d}{2}\rceil$ parameters subject to the normalization condition

\be
\sum_{j=0}^{\lfloor \frac{d-1}{2}\rfloor} {r_j} 
+\sum_{j=\lfloor \frac{d-1}{2}\rfloor+1}^{d-1} {r_{d-1-j}}=1.
\ee

\end{observation}

Unfortunately, when the energy constraint \eqref{Econstr} is employed, the symmetry \eqref{eq:coefsym} can no longer be exploited. We do not have the freedom to impose the same coefficients along a $k$-diagonal of the density matrix due to the energy constraint. 

\begin{fact}\label{property2}
The states $\ket \varphi$ and $V\ket\varphi$ with
\be
V:=\sum_{i=0}^{d-1}\ket i \bra{d-1-i} ,
\ee 
perform the same for discriminating dephasing channels. 
\end{fact}
\begin{proof}
Setting $\varphi=\ket{\varphi}\bra{\varphi}$ and
\be
\label{eq:Delta}
\Delta\left(\gamma_0,\gamma_1; \ket\varphi \right)\equiv\cN_{\gamma_0} \left(\varphi\right)- \cN_{\gamma_1} \left(\varphi\right), 
\ee
one can easily show that
\be
\Delta\left(\gamma_0,\gamma_1;V \ket\varphi \right)=V\Delta\left(\gamma_0,\gamma_1; \ket\varphi \right)V^\dag,
\ee
hence, 
\be
\left\|\Delta\left(\gamma_0,\gamma_1; \ket\varphi \right)\right\|_1=\left\|\Delta\left(\gamma_0,\gamma_1;V \ket\varphi \right)\right\|_1
\ee
Therefore, the states $\ket \varphi$ and $V\ket\varphi$ perform the same for discriminating channels $\cN_{\gamma_0} $ and $\cN_{\gamma_1} $.
\end{proof}
 
 \begin{observation}
As a consequence of Fact \ref{property2}, we can conclude that if a state $\ket{\varphi}=
\sum_{j=0}^{d-1} r_j \ket{j}$ is optimal with energy constraint $E$, then the state $V\ket{\varphi}$
will be optimal with energy constraint $d-1-E$. This means that we can restrict the analysis to half of the energy range,  i.e. $0\leq E\leq \frac{d-1}{2} $.
\end{observation}


\section{Discrimination without energy constraint}
\label{sec:DnoE}
We want to discriminate between two dephasing channels $\cN_{\gamma_0}$ and $\cN_{\gamma_1}$ with dephasing parameters $\gamma_0$, $\gamma_1$ respectively,
 and each appearing with probability $\frac{1}{2}$. From here on, without loss of generality, we assume $\gamma_0<\gamma_1$ due to the symmetry of the trace norm w.r.t. $\gamma_0$ and $\gamma_1$. 

\subsection{The qubit case}\label{sec:qubit}

Following Sec.\ref{sec:pre}, for  $d=2$ we would consider
\begin{equation}
    \label{eq:Psi01}
    \ket{\varphi}=\sqrt{r_0}\ket{0}+\sqrt{r_1}\ket{1},
\end{equation}
where $r_0,r_1\in\mathbb{R}_0^+$, such that
\bea\label{normqubit}
r_0+r_1=1.
\eea
Furthermore, referring to Observation \ref{ob3} we have $r_0=r_1$, which leads us to $r_0=r_1=1/2$.

It is then easy to see that
\be\label{Psfqubit}
\overline{P}_s=\frac{1}{2}\left(1+\frac{1}{2}g_1\right),
\ee
where
\begin{equation}\label{adef}
   g_1\equiv e^{-\gamma_0/2}-e^{-\gamma_1/2}.
\end{equation}

\subsection{The qutrit case}\label{sec:qutrit}

We study here the case $d=3$.
Following Sec.\ref{sec:pre} we would consider
\begin{equation}\label{inputqutrit}
    \ket{\varphi}=\sqrt{r_0}\ket{0}+\sqrt{r_1}\ket{1}+\sqrt{r_2}\ket{2},
\end{equation}
where $r_0,r_1,r_2\in\mathbb{R}_0^+$, such that 
\be
r_0+r_1+r_2=1.
\ee
Additionally, referring to Observation \ref{ob3}, we have $r_0=r_2$.
Thus, the normalization condition becomes
\begin{equation}\label{normqutrit}
2r_0+r_1=1.
\end{equation}
After writing $r_1=1-2r_0$, we get (w.r.t. the basis $\cal B$)
\be \label{Deltaqutrit}
\Delta(\gamma_0,\gamma_1;\ket{\varphi})
=\begin{pmatrix}
0   &  g_1 \sqrt{r_0(1-2r_0)} & g_2 r_0\\
g_1 \sqrt{r_0(1-2r_0)} & 0 &  g_1 \sqrt{r_0(1-2r_0)} \\ 
g_2 r_0  &  g_1 \sqrt{r_0(1-2r_0)}  &  0
\end{pmatrix},
\ee
where, in addition to \eqref{adef}, we have set
\begin{equation}\label{bdef}
g_2\equiv e^{-4\gamma_0/2}-e^{-4\gamma_1/2}.
\end{equation}
Eq.\eqref{Deltaqutrit} leads to
 \be\label{Deltaqu3}
 \| \Delta(\gamma_0,\gamma_1;\ket{\varphi}) \|_1=\sqrt{r_0[g_1^2 \left(8-16 r_0\right)+g_2^2 r_0]}+g_2 r_0,
 \ee
 which is the quantity to be optimized in terms of $r_0$. Using the basic calculus, 
 we arrive at the following results.
 \begin{itemize}
 \item[i)] When $2 g_1>g_2$
 \be
\overline{ P_{s}}=\frac{1}{2}\left(1+\frac{{2}g_1^2}{4g_1-g_2}\right),
 \ee
 with the optimal input state \eqref{inputqutrit} having
 \be\label{t0t2max}
r_0=r_2=\frac{g_1}{4 g_1-g_2}, \quad\quad r_1=\frac{2 g_1-g_2}{4 g_1-g_2}.
 \ee
 
 \item[ii)] When $2g_1\leq g_2$ 
  \be
\overline{ P_{s}}=\frac{1}{2}\left(1+\frac{g_2}{2}\right),
 \ee
  with the optimal input state \eqref{inputqutrit} having
 \be\label{t0t2min}
r_0=r_2=\frac{1}{2}, \quad\quad r_1=0.
 \ee
 \end{itemize}
 
 This shows that the parameters space $\{\gamma_0,\gamma_1\}$,
 besides being symmetric w.r.t. $\gamma_0=\gamma_1$, 
 it is further divided into two parts by a line $2 g_1=g_2$ for which $\gamma_0\neq\gamma_1$ (see Fig.\ref{figDIF}). {Notice that the piecewise function $\overline{P_s} $ is continuous w.r.t. to the variables $\gamma_0$ and $\gamma_1$, however its first derivatives are discontinuous in the line $2g
 _1=g_2$.} 
\begin{figure}[H]
\centering
\includegraphics[width=9cm]{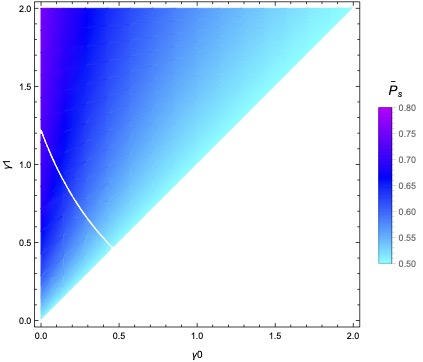}
\caption{{Density plot of $\overline{P}_s$ in the parameters space $\{\gamma_0,\gamma_1\}$ with highlighted the two regions
i) $2 g_1>g_2 $ (bottom left) and ii) $2 g_1\leq g_2$ (top right).} Here and in the plots below, the part $\gamma_0> \gamma_1$ is 
not reported because of the symmetry w.r.t. $\gamma_1=\gamma_0$.}
\label{figDIF}
\end{figure}

\subsection{The ququart case}\label{sec:ququart}

When $d=4$, following Sec.\ref{sec:pre} we would consider
\begin{equation}\label{inputqudit}
    \ket{\varphi}=\sqrt{r_0}\ket{0}+\sqrt{r_1}\ket{1}+\sqrt{r_2}\ket{2}+\sqrt{r_3}\ket{3},
\end{equation}
where $r_0,r_1,r_2,r_3\in\mathbb{R}_0^+$, such that 
\be\label{normququart}
r_0+r_1+r_2+r_3=1.
\ee
Additionally, we can confine the optimization to input
state of the form \eqref{inputqudit} with $r_2 = r_1$ and $r_3 = r_0$ by referring to Observation \ref{ob3}. Thus, the normalization condition becomes
\begin{equation}\label{normqudit}
r_0+r_1=\frac{1}{2}.
\end{equation}
After writing {${r_1}={\frac{1-2 r_0}{2}}$} we get (w.r.t. the basis $\cal B$) 
\begin{align}\label{Deltaququart}
&\Delta(\gamma_0,\gamma_1;\ket{\varphi})\notag\\
&=\begin{pmatrix}
0   &  g_1 \sqrt{r_0\left(\frac{1}{2}-r_0\right)} & g_2  \sqrt{r_0\left(\frac{1}{2}-r_0\right)} 
& g_3\left(\frac{1}{2}-r_0\right) \\
g_1 \sqrt{r_0\left(\frac{1}{2}-r_0\right)}   & 0 &  g_1 r_0& g_2 \sqrt{r_0\left(\frac{1}{2}-r_0\right)}  
\\ 
g_2 \sqrt{r_0\left(\frac{1}{2}-r_0\right)} & g_1 r_0 & 0 & g_1 \sqrt{r_0\left(\frac{1}{2}-r_0\right)} 
\\
g_3\left(\frac{1}{2}-r_0\right) & g_2 \sqrt{r_0\left(\frac{1}{2}-r_0\right)} &  g_1 \sqrt{r_0\left(\frac{1}{2}-r_0\right)}   &  0
\end{pmatrix},
\end{align}
where, in addition to \eqref{adef} and \eqref{bdef}, we have set
\begin{equation}
g_3\equiv e^{-9\gamma_0/2}-e^{-9\gamma_1/2}.
\end{equation}
Eq.\eqref{Deltaququart} leads to
 \be\label{1normqu4}
 \| \Delta(\gamma_0,\gamma_1;\ket{\varphi}) \|_1=\frac{1}{4} \left(|\zeta+\sqrt{\xi_+}|+|\zeta-\sqrt{\xi_+}| +|\zeta+\sqrt{\xi_-}|+|\zeta-\sqrt{\xi_-}|\right),
 \ee
 where
\begin{align}
\label{eq:zetas}
&\zeta \equiv g_3+2(g_1-g_3)r_0,\\
&\xi_\pm \equiv g_3^2+4  \left(2(g_1\pm g_2)^2-g_3 (g_1+g_3)\right)r_0+4((g_1+g_3)^2-4(g_1\pm g_2)^2 )r_0^2.
\end{align}

Notice that $\xi_\pm$ are parabolas in terms of $r_0 \in \left[0,\frac{1}{2}\right]$ and $\zeta, \xi_-,\xi_+$ are always non-negative in the region $r_0 \in \left[0,\frac{1}{2}\right]$.
Furthermore, we have
\bea
&\sqrt{\xi_+}>\zeta>\sqrt{\xi_-}, \quad (g_1-g_2)^2 {<} g_1 g_3,\\
&\sqrt{\xi_+}>\sqrt{\xi_-}>\zeta, \quad (g_1-g_2)^2{>}g_1 g_3.
\eea
Taking into account these properties we arrive at
 \be
\| \Delta(\gamma_0,\gamma_1;\ket{\varphi}) \|_1=\left\{\begin{array}{lll}
   \frac{\zeta+\sqrt{\xi_+}}{2},       & &   (g_1-g_2)^2{\leq} g_1 g_3\\
\frac{\sqrt{\xi_-}+\sqrt{\xi_+}}{2},     & & (g_1-g_2)^2{>}  g_1 g_3
\end{array}.\right.
 \ee
 Then, the maximum probability of success and the optimal input states result as follows.

 \begin{itemize}
 \item[i)] When $ (g_1-g_2)^2\leq g_1 g_3 $, we have:
 
 \begin{enumerate}
 {
\item[i.1)] If $g_1+g_2\leq g_3$, it is
\be
\overline P_s=\frac{1}{2}\left(1+\frac{g_3}{2}  \right),
\ee
with the optimal input state having
\be
r_0=r_3=\frac{1}{2},\quad\quad r_1=r_2=0.
\ee
  \item[i.2)]
   If $g_3<g_1+g_2$, it is
\be
\overline P_s=\frac{1}{2}\left(1+\frac{(g_1+g_2)^2-g_1 g_3}{2(g_1+2 g_2-g_3)}  \right),
\ee
with the optimal input state having
\be
r_0=r_3=\frac{g_2}{2(g_1+2 g_2-g_3)},\quad \quad r_1=r_2=\frac{g_1+g_2-g_3}{2(g_1+2 g_2-g_3)}.
\ee

}
\end{enumerate}

 \item[ii)] When $ (g_1-g_2)^2>g_1 g_3 $, we have:
 \be\label{Psququart2}
 \overline P_s=
 \frac{1}{2}\left(1+\frac{\sqrt{\xi_-}+\sqrt{\xi_+}}{4}\right),  
 \ee
  with the optimal input given by
 \be
 r_0=r_3=1/2-t^*, \quad \quad  r_1=r_2=t^*.
 \ee
  {where $t^*$ is defined in the following subclasses ii.1 and ii.2}
 \begin{enumerate}
 \item[ii.1)] $4 (g_1 - g_2)^2 < (g_1 + g_3)^2$
 \be
t^*=\max\{ t_+^*, t_-^*\},
 \ee

 where 
\begin{align}\label{eq:tstarpm}
 t^*_\pm \equiv \frac{\beta }{\alpha }{\pm} \sqrt{\frac{pQ+\Xi+Q^2}{48\alpha Q}}
 +\sqrt{
 \frac{2pQ-\Xi-Q^2}{48\alpha  Q}
 \mp q\sqrt{\frac{3\alpha Q}{pQ+\Xi+Q^2}}}.
\end{align}
The explicit expressions of the symbols are reported in Appendix \ref{app:symbols}.

 \item[ii.2)] $ 4 (g_1 - g_2)^2 \geq (g_1 + g_3)^2$
\be
t^*=t^*_+.
 \ee

\end{enumerate}

 \end{itemize}
 
 This shows that the parameters space $\{\gamma_0, \gamma_1\}$, besides being symmetric w.r.t.
$\gamma_0 = \gamma_1$, it is further divided into four parts by lines $g_1+g_2=g_3$, $(g_1-g_2)^2=g_2 g_3$ and $4(g_1-g_2)^2=(g_1+g_3)^2$ (see Fig.\ref{rplot2}).

\begin{figure}[H]
\centering
\includegraphics[width=9cm]{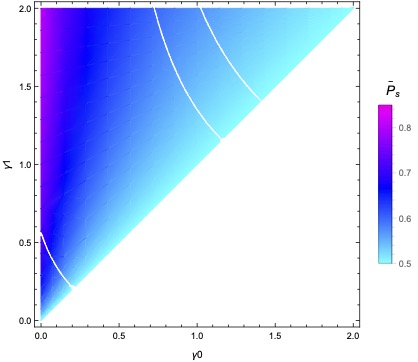}
\caption{{Density plot of $\overline{P}_s$ in the parameters space $\{\gamma_0,\gamma_1\}$ with highlighted the four regions (from bottom left corner to top right corner): i.1) $g_1+g_2\leq g_3 \wedge (g_1-g_2)^2\leq g_1 g_3$; i.2) $g_1+g_2>g_3 \wedge (g_1-g_2)^2 \leq g_1 g_3$;} ii.1) $(g_1-g_2)^2> g_1 g_3\wedge 4(g_1-g_2)^2\leq (g_1+g_3)^2$ 
and ii.2) $(g_1-g_2)^2> g_1 g_3\wedge 4(g_1-g_2)^2>(g_1+g_3)^2$.}
\label{rplot2}
\end{figure}
 
{The piecewise function $\overline{P_s} $ is continuous w.r.t. to the variables $\gamma_0$ and $\gamma_1$, however its first derivatives are discontinuous in the lines $g_1+g_2=g_3$, $(g_1-g_2)^2=g_2 g_3$ and $4(g_1-g_2)^2=(g_1+g_3)^2$ }


\section{Discriminating with energy constraint}

In this section, we will solve the optimization problem \eqref{Psf} subject to the constraint \eqref{Econstr} for $d=2,3,4$.
Taking into account the Observation \ref{ob3}, it is enough to consider $0<E\leq \frac{d-1}{2}$.

\subsection{The qubit case}
Imposing the energy constraint \eqref{Econstr} to the state \eqref{eq:Psi01},
we have 
\begin{equation}
E=r_1. 
\end{equation}

Expressing $r_0=1-E$, thanks to \eqref{normqubit}, we will then arrive to 
\begin{equation}
 \label{eq:PsPsi01}
 \overline{P}_{s}=\frac{1}{2}\left(1+g_1 \sqrt{(1-E)E}\right),
\end{equation}
which coincide with \eqref{Psfqubit} for $E=1/2$.


\subsection{The qutrit case}\label{sec:d3E}
If we now impose the energy constraint \eqref{Econstr} to the state \eqref{inputqutrit},
the problem becomes of finding
\begin{align}\label{t2constr}
{\overline{P}_s=}\max_\varphi  P_s(\varphi)  
\quad \text{s. t.}\quad &\begin{array}{cc} 
r_1+2r_2&=E\\
r_0+r_1+r_2&=1
\end{array}.
\end{align}
The {numerical} results, namely $\overline{P}_s$ for various energy constraints, are shown in Fig.\ref{figP3}.  We may see that when the energy is smaller than 0.5, which is the optimal energy for qubit, there always exists a region where qubit performs better than qutrit. Such a region shrinks as energy increases, and nullifies when $E= 0.5$. In other words, qutrit always performs better than qubit when $E\geq 0.5$. 
This is obtained, for a fixed value of $E$, by comparing the values of $\overline{P}_s$ coming from \eqref{t2constr} with those from  \eqref{eq:PsPsi01} on each point of the parameters space.

\begin{figure}[H]
    \centering
    \begin{subfigure}[t]{0.5\textwidth}
        \centering
        \includegraphics[width=\textwidth]{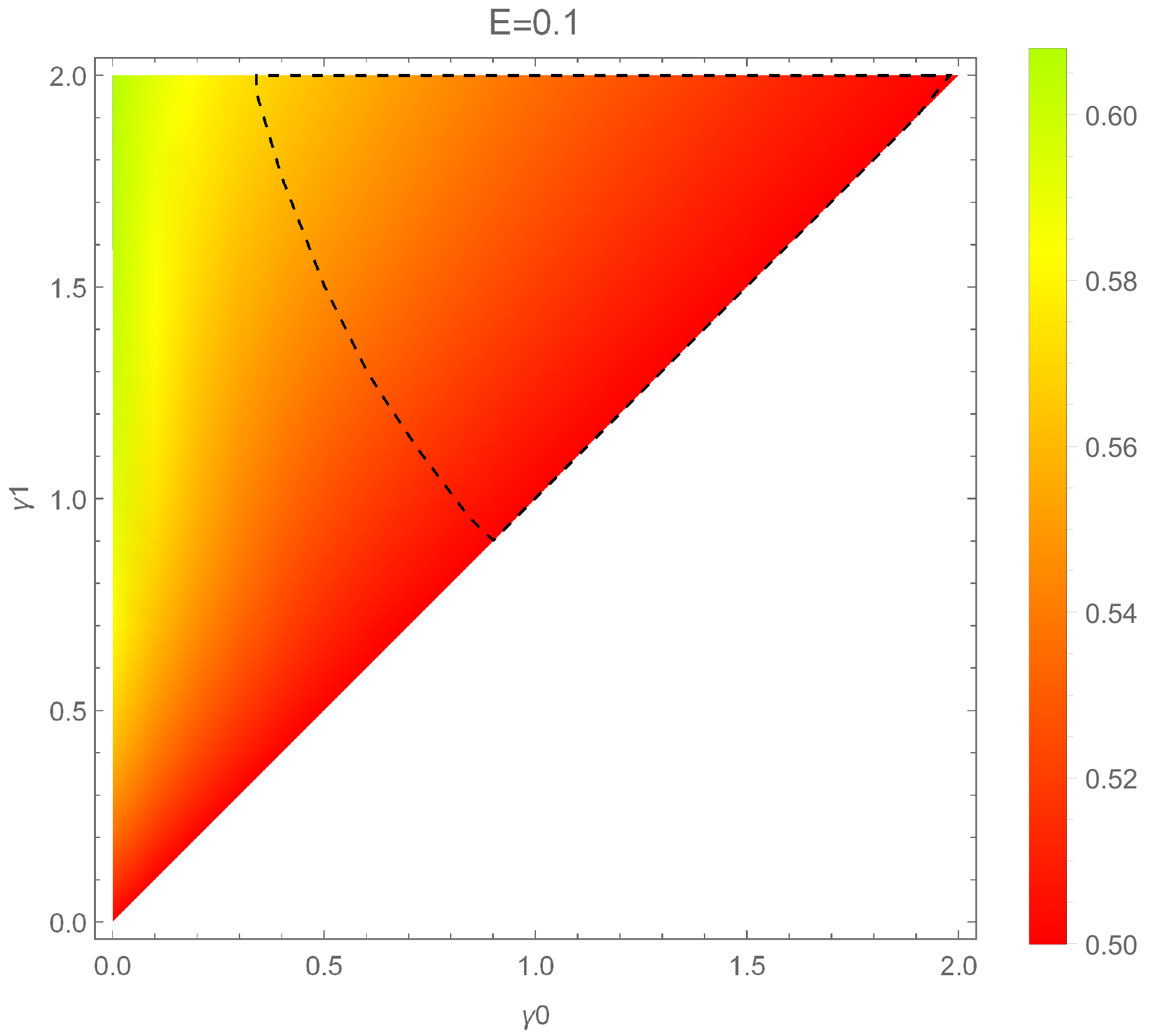}
    \end{subfigure}%
   \hfill
    \begin{subfigure}[t]{0.5\textwidth}
        \centering
        \includegraphics[width=\textwidth]{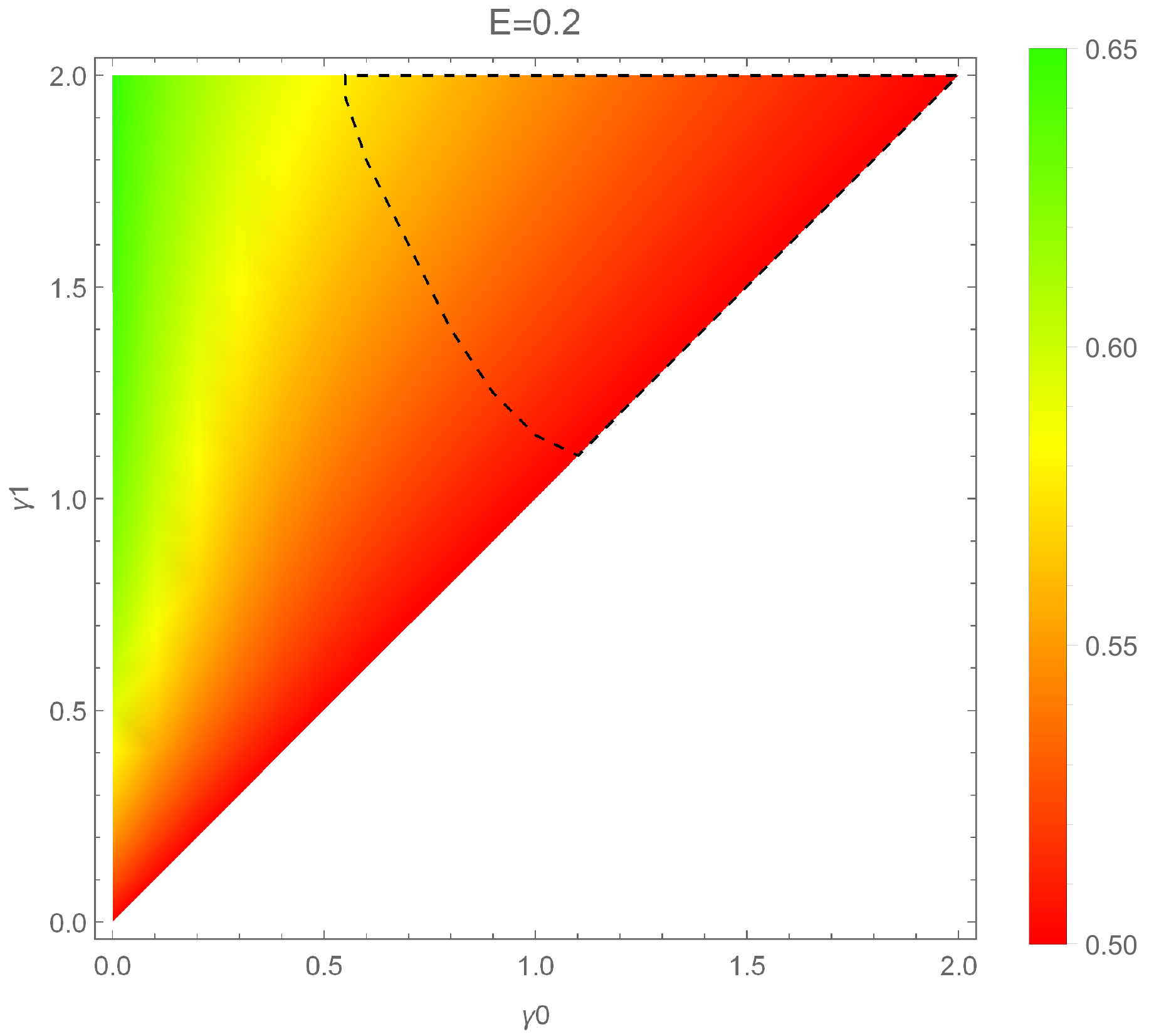}
    \end{subfigure}
    \newline 
    \begin{subfigure}[t]{0.5\textwidth}
        \centering
        \includegraphics[width=\textwidth]{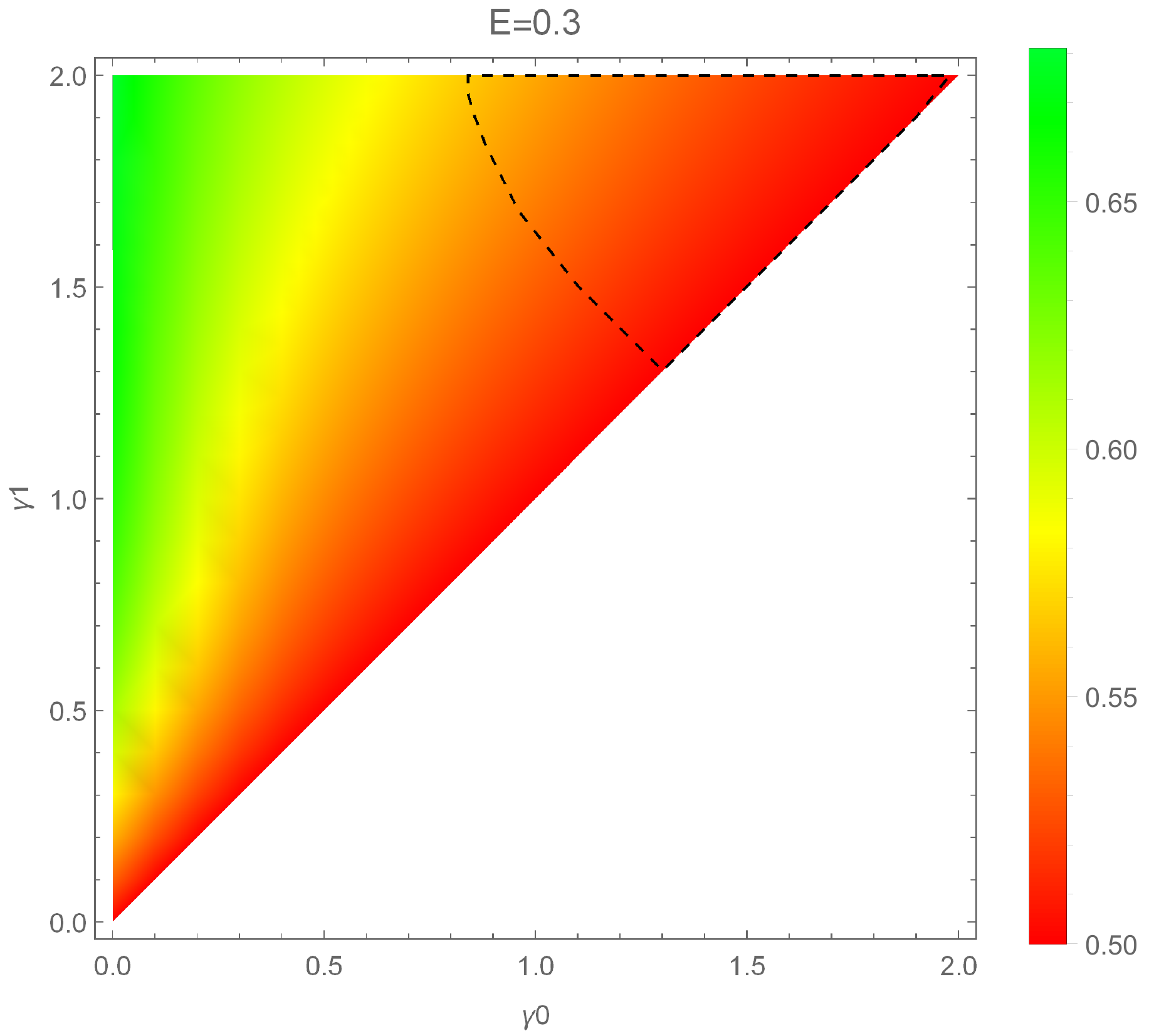}
    \end{subfigure}%
   \hfill
    \begin{subfigure}[t]{0.5\textwidth}
        \centering
        \includegraphics[width=\textwidth]{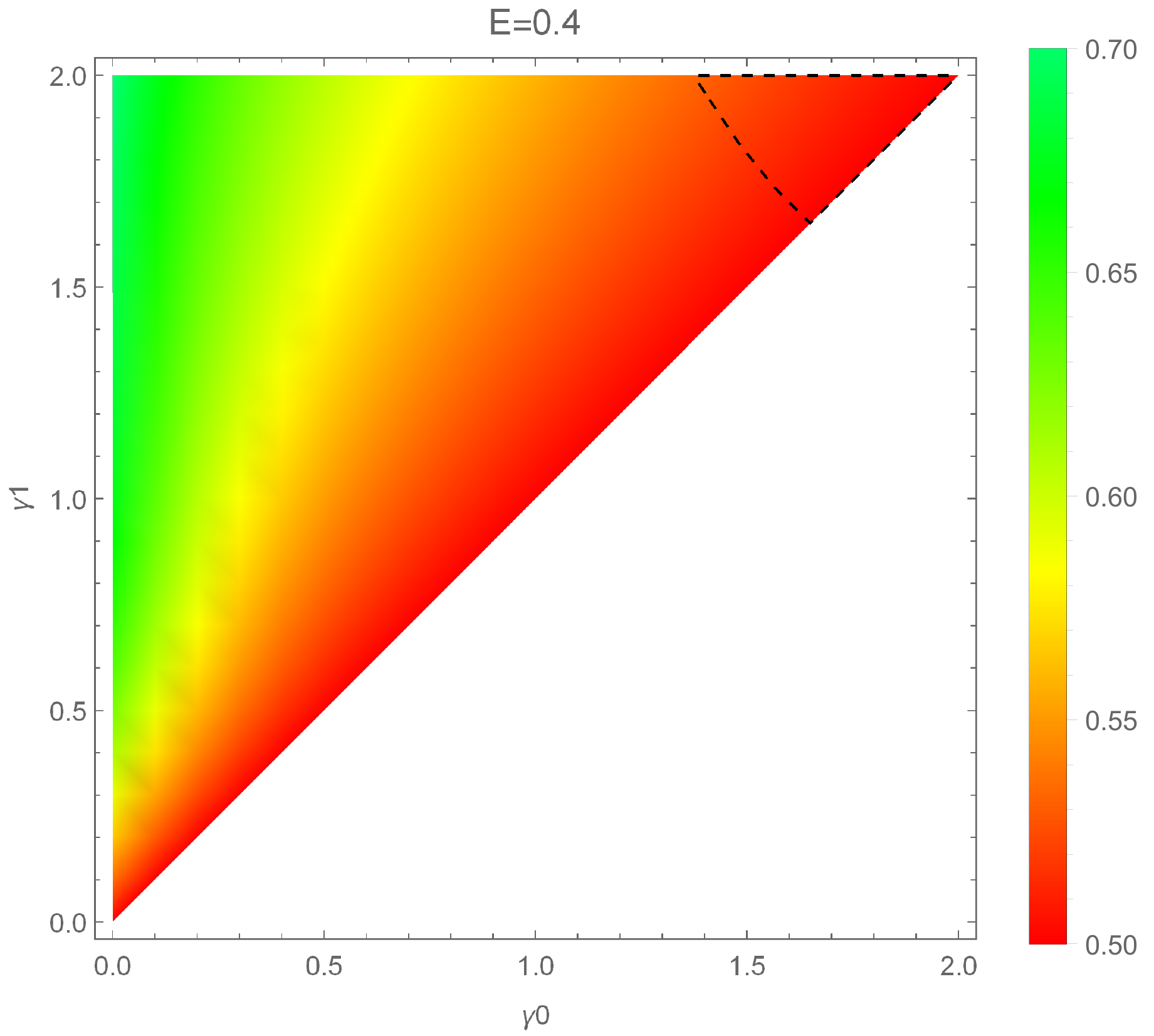}
    \end{subfigure}
    \caption{Density plots of $\overline{P}_s$ in the parameter region $\{\gamma_0,\gamma_1\}$  corresponding to different energy constraints. The black dashed line marks the {border} of the region where the optimal input state becomes a qubit state (top right corner).}  
     \label{figP3}
\end{figure}


\subsection{The ququart case}

After imposing the energy constraint \eqref{Econstr} to the state  \eqref{inputqudit}, the problem becomes of finding
\begin{align}\label{d4constr}
\overline{P}_s=&\max_\varphi  P_s(\varphi)  
\quad \text{s. t.}\quad &\begin{array}{cc} 
r_1+2r_2+3r_3&=E\\
r_0+r_1+r_2+r_3&=1
\end{array},
\end{align}
The {numerical} results, namely $\overline{P}_s$ for various energy constraints, are shown in Fig.\ref{figP4}.  We may see that when $E<1$, there always exists a region where qutrit is the optimal. This region shrinks as $E$ increases, vanishes when $E=1$ which is the optimal energy level for $d=3$. 
This is obtained, for a fixed value of $E$, by comparing the values of $\overline{P}_s$ coming from
\eqref{d4constr} with those from \eqref{t2constr} on each point of the parameters space.

\begin{figure}[H]
    \centering
    \begin{subfigure}[t]{0.5\textwidth}
        \centering
        \includegraphics[width=\textwidth]{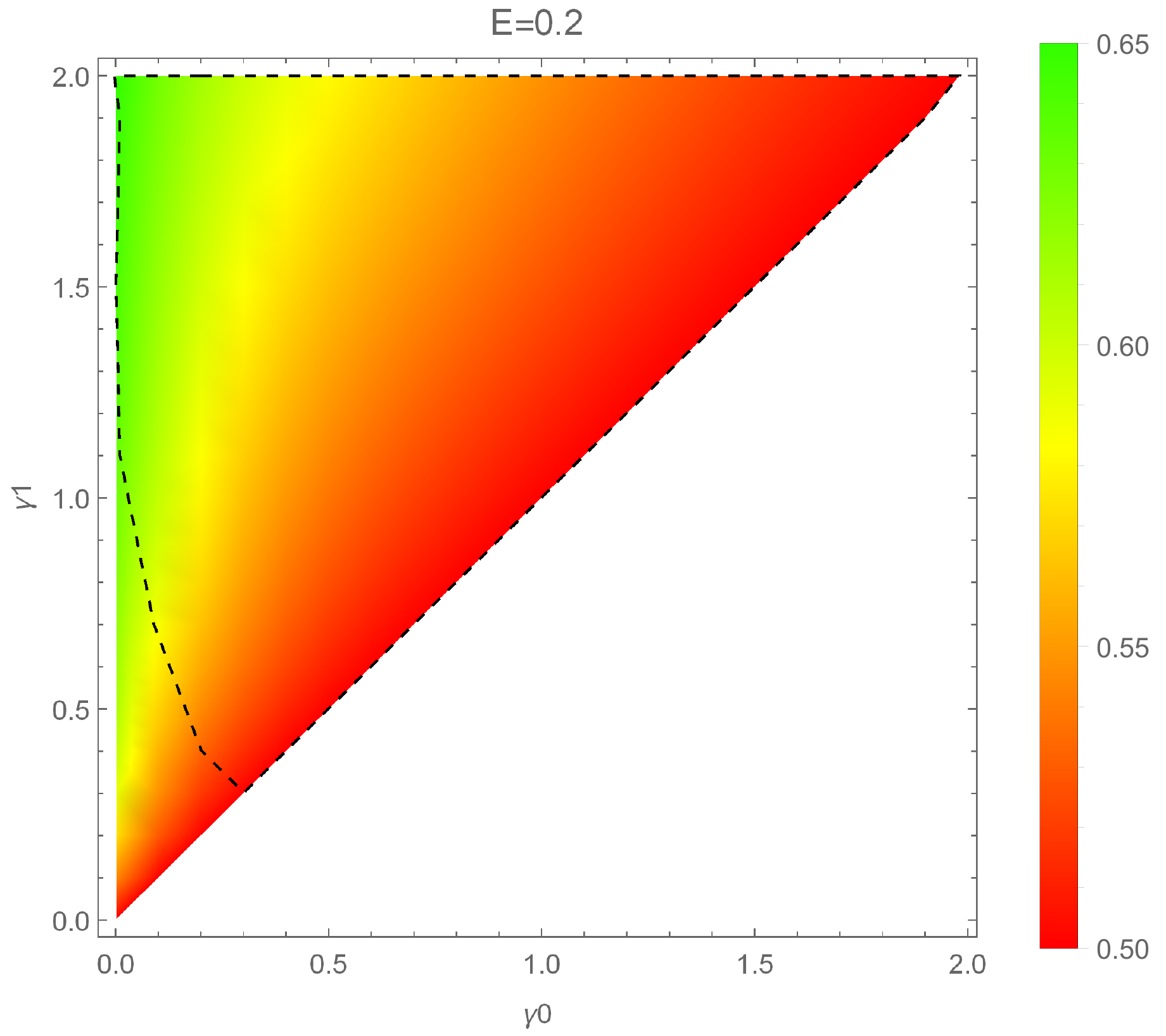}
    \end{subfigure}%
   \hfill
    \begin{subfigure}[t]{0.5\textwidth}
        \centering
        \includegraphics[width=\textwidth]{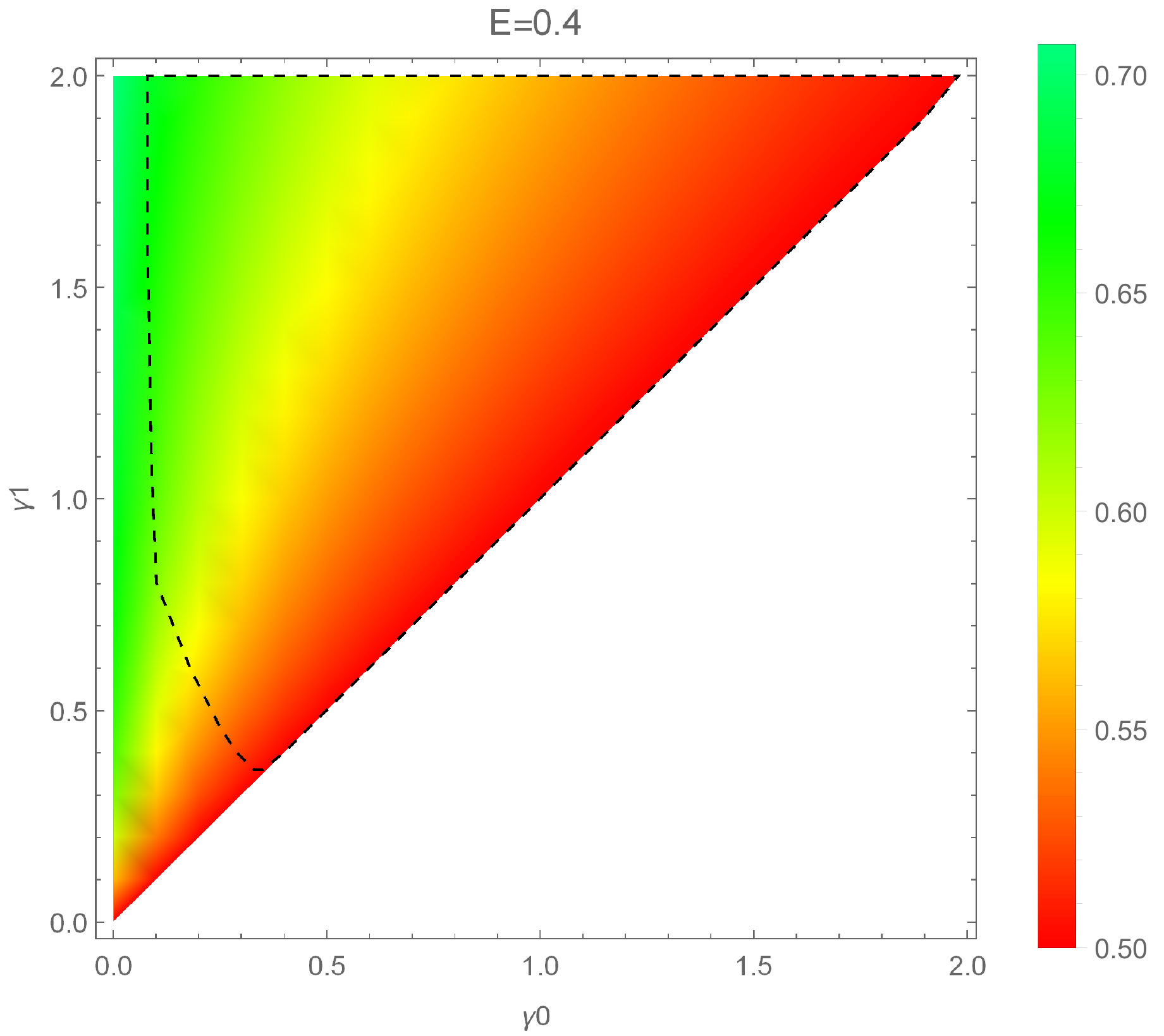}
    \end{subfigure}
    \newline 
    \begin{subfigure}[t]{0.5\textwidth}
        \centering
        \includegraphics[width=\textwidth]{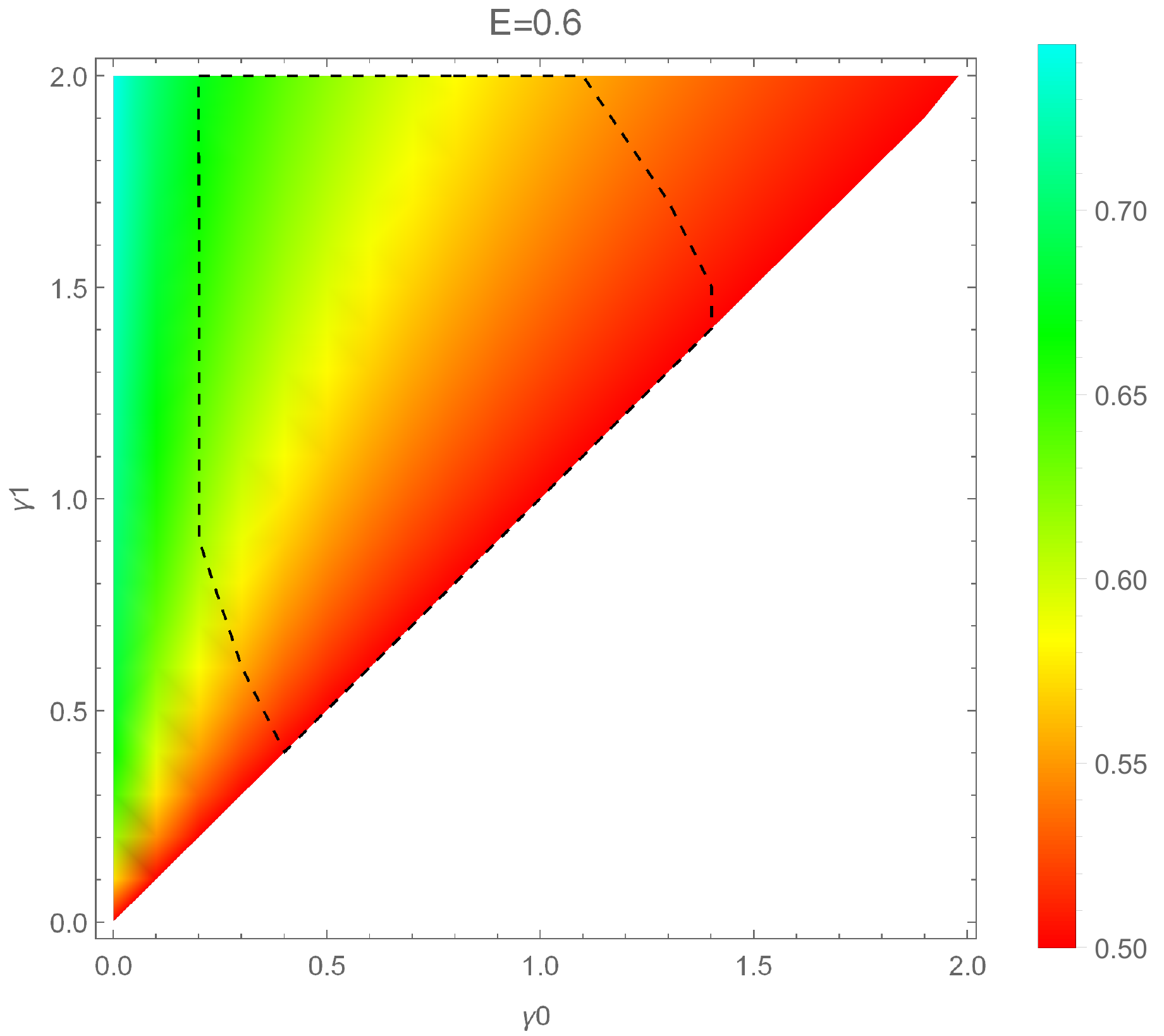}
    \end{subfigure}%
   \hfill
    \begin{subfigure}[t]{0.5\textwidth}
        \centering
        \includegraphics[width=\textwidth]{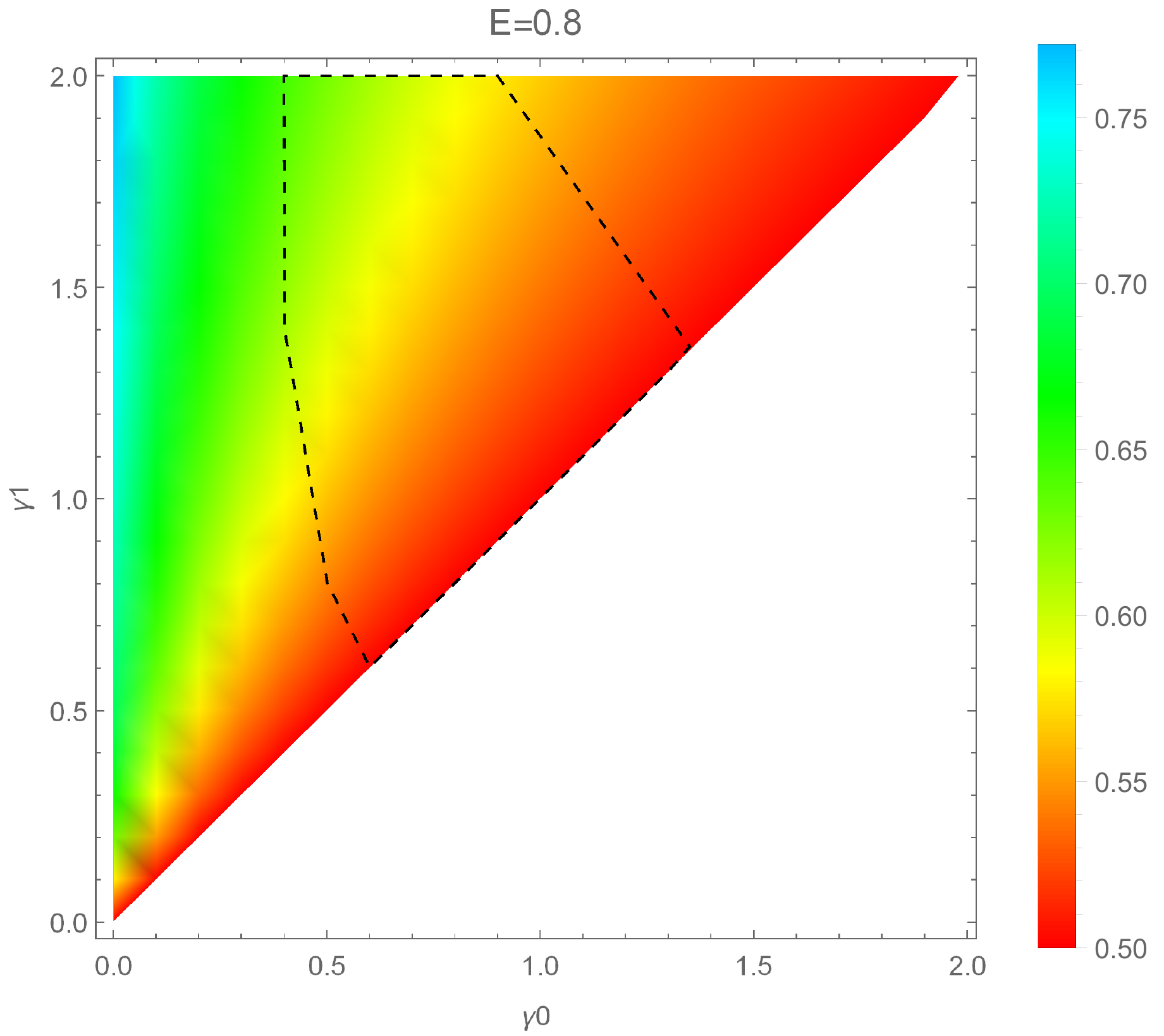}
    \end{subfigure}
    \caption{Density plots of $\overline{P}_s$ in the parameter region $\{\gamma_0,\gamma_1\}$  corresponding to different energy constraints. The black dashed line marks the {border} of the region where the optimal input state becomes a qutrit state.
}
    \label{figP4}
\end{figure}


\section{Discrimination with side entanglement}

In this Section we show that the celebrated resource of entanglement is not helping the discrimination of phase damping channels. {This comes into play when considering the tensor product of identity map with phase damping channel \footnote{{The output entropy of such kind of composition was estimated in Refs.\cite{Amosov2}}}.}

The standard process of channel discrimination with side entanglement involves the use of
a maximally entangled state (MES) $\psi=|\psi\rangle\langle\psi|$ in ${\cal H}\otimes {\cal H}$ whose half 
is sent through the channel so that at the output the measurement can be done on the global state 
\be
({\rm id}\otimes {\cal N}_{\gamma})(\psi).
\ee
Then the problem is how to contemplate all maximally entangled states as probes.

Given an orthonormal basis $\{|e_i\rangle\}_{i=0}^{d-1}$ for $\cal H$, e.g. the canonical one, a MES can be written as 
\begin{equation}
|\psi\rangle=\frac{1}{\sqrt{d}}\sum_{i=0}^{d-1}|e_i\rangle|e_i\rangle.
\end{equation}
{
Although there are infinitely many bases that we can use to construct MES,
 it could be enough to consider those bases that mutually differ as much as possible, namely the Mutually Unbiased Bases (MUB) \cite{Berge,Eusebi}.

However, we will follow hereafter another route that will also encompass non maximally entangled states.

A generic two-qudit state can be written, in the Schmidt form, as 
\be\label{qd}
\sum_{i=0}^{d-1}\sqrt{r_i} |e_i\rangle |f_i\rangle,
\ee
where $\{\ket{e_i}\}_{i=0}^{d-1}$, $\{\ket{f_i}\}_{i=0}^{d-1}$ are two orthonormal bases and $r_i\geq 0$ 
with $\sum_{i=0}^{d-1}r_i=1$.
Since we are interested on the benefit that entanglement can provide, 
we can equivalently consider 
\be\label{qdd}
|\psi\rangle=\sum_{i=0}^{d-1}\sqrt{r_i} |i\rangle |i\rangle,
\ee
obtained from \eqref{qd} by applying local unitaries. Then, taking $\psi=|\psi\rangle\langle\psi|$ and sending half of this state through the channel,  
it turns out that 
\be
({\rm id}\otimes {\cal N}_{\gamma})(\psi)= 
\sum_{i} r_i |ii\rangle\langle ii|
+\sum_{i\neq j} \sqrt{r_ir_j} e^{-\frac{\gamma}{2}(i-j)^2} |ii\rangle\langle jj|.
\ee
As a consequence, we can express
$({\rm id}\otimes {\cal N}_{\gamma_0})(\psi)
-({\rm id}\otimes {\cal N}_{\gamma_1})(\psi)$ in a $d\times d$ matrix restricting to
the subspace spanned by $\{\ket{ii}\}_{i}$. 
This would be the same matrix as for the case without side entanglement.
Hence, taking into account that the arguments leading to \eqref{eq:coefsym} still hold, 
the quantity
\be
\left\|
({\rm id}\otimes {\cal N}_{\gamma_0})(\psi)
-({\rm id}\otimes {\cal N}_{\gamma_1})(\psi)
\right\|_1 
\ee
gives, upon optimization over the $r_i$s, a result identical to that without entanglement.
Thus, the inutility of side entanglement is affirmed in any dimension $d$.
}


\section{Concluding remarks}

In conclusion, we studied the optimization problem for the success probability in discriminating two dephasing channels, with and without input energy constraint.
Analytical solutions for the lower dimensions ($d=2,3,4$) without input energy constraint are found. It is also shown the inutility of side entanglement strategy.

In all cases analyzed,  
the optimal constrained success probability coincides with the optimal unconstrained one, 
not for the maximum value of the energy ($E=1,2,3$ in the qubit, qutrit and ququart case respectively), but for half of it. However, this still implies a linear increment of the energy with the dimension of the system.
Thus, {we expect that} in infinite dimensional systems the optimal states would be unphysical.
On the other hand, with an energy constraint, the optimal states turn out 
to live in {a space whose dimension clearly increases with the value of $E$
as summarized in Fig.\ref{figch} (taking into account the results of Figs.\ref{figP3} and \ref{figP4}).} 
The highest dimension {of Hilbert space} is always necessary on the bottom-left corner, where $\gamma_1\gamma_0\rightarrow0$.
\begin{figure}[H]
    \centering
    \begin{subfigure}[t]{0.5\textwidth}
        \centering
        \includegraphics[width=\textwidth]{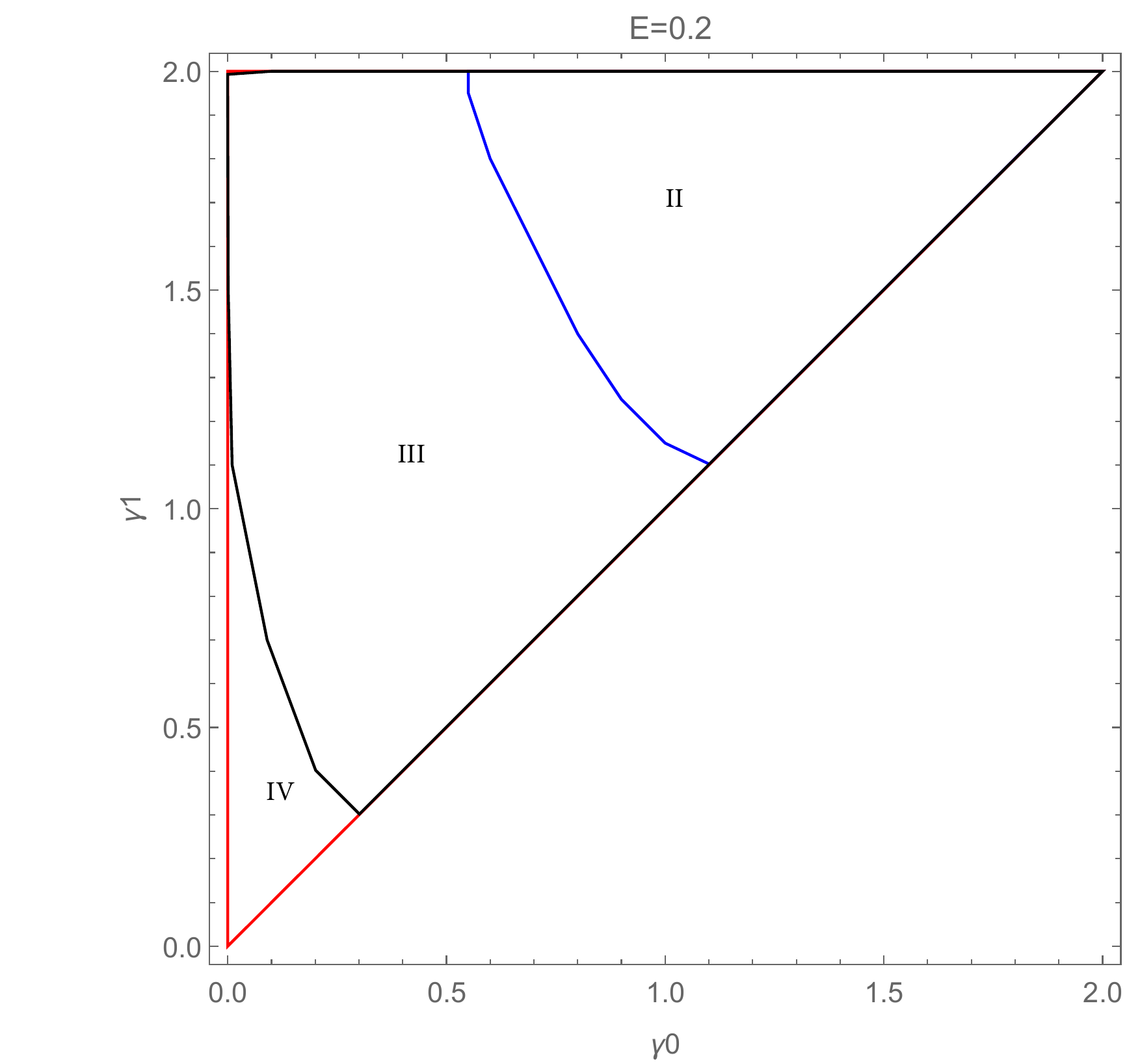}
    \end{subfigure}%
   \hfill
    \begin{subfigure}[t]{0.5\textwidth}
        \centering
        \includegraphics[width=\textwidth]{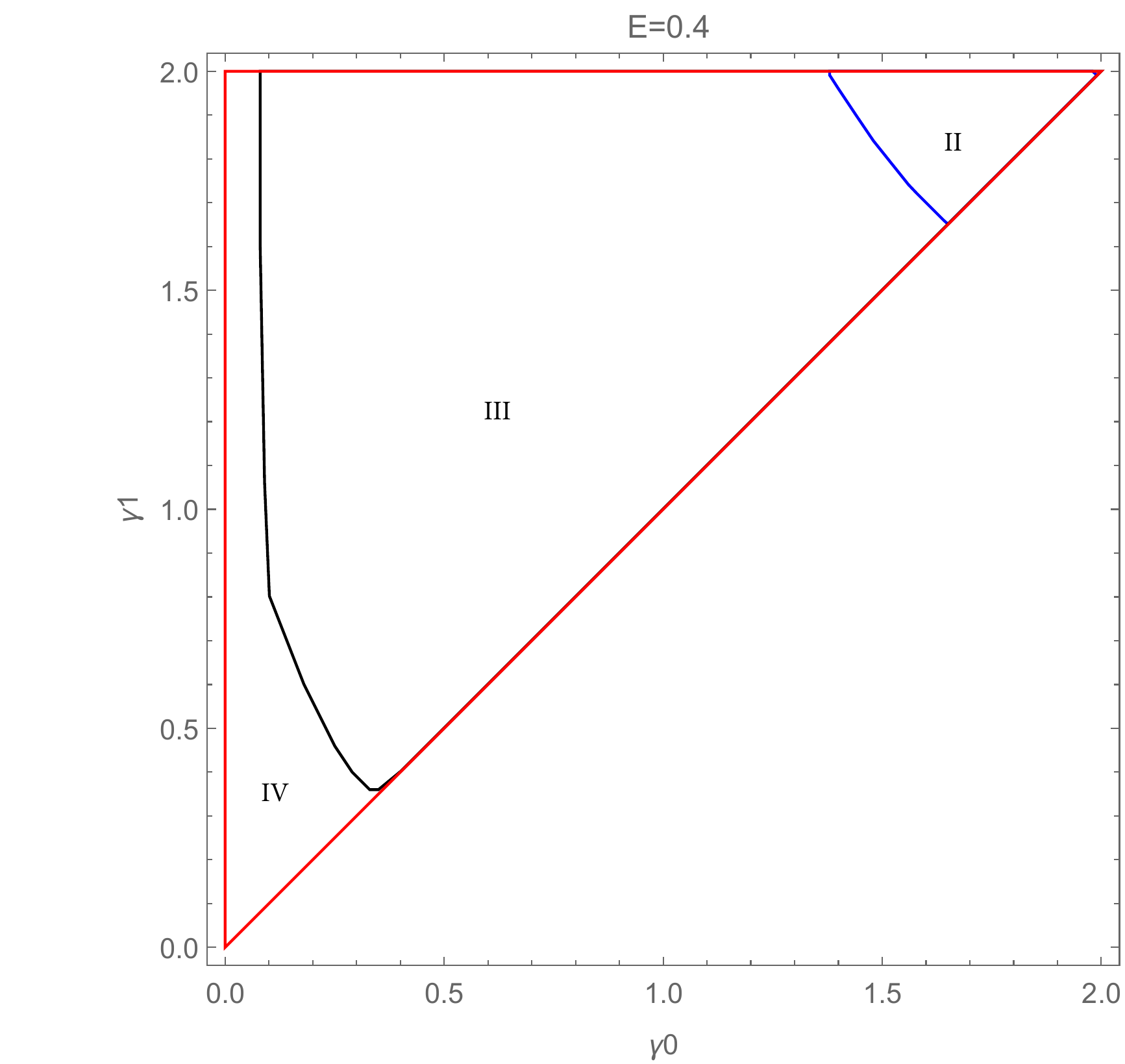}
    \end{subfigure}
 
    \caption{{Parameters region $\{\gamma_0,\gamma_1\}$ divided into parts according to the optimal dimension of the input state for a fixed energy constraint.}
    Labels II, III, IV refer to dimensions $d=2,3,4$ respectively.}
    \label{figch}
\end{figure}

About the discrimination with side entanglement, it is worth observing the following:
i) Side entanglement is very useful for Pauli channels \cite{Max};
ii) Side entanglement is partially useful for amplitude damping \cite{MilaJPA};
iii) Side entanglement is not useful for phase damping channel (present work).
Then, we are led to {conjecture} that side entanglement is useful in channel discrimination if 
the Kraus operators of the channel admit a polar decomposition with non-trivial unitaries {(i.e., not identities)}.
Actually, if the Kraus operators are unitaries, then there is maximal benefit from side entanglement.
In contrast, if all Kraus operators have simply the identity as unitary part in the polar decomposition, there is no benefit from side entanglement.
This aspect will be deepen in the future work.
\\

\bigskip

The authors equally contributed to this work.
We acknowledge the funding from the European Union's
Horizon 2020 research and innovation programme under Grant agreement
No 862644 (FET-Open project: QUARTET).
\\


 



\appendix

\section{Symbols of Eq.\eqref{eq:tstarpm}}\label{app:symbols}

The symbols are defined as
 \begin{align}
\alpha& \equiv 2 \left[(g_1-g_3)^2 -4 g_2^2\right] \left[(3 g_1+g_3)^2-4 g_2^2\right], \\
\beta &\equiv g_1^2(6 g_1^2-7 g_1 g_3-18 g_2^2-3 g_3^2)
-g_1(8 g_2^2 g_3-3 g_3^3)+8 g_2^4-6 g_2^2 g_3^2+g_3^4,\\
p&\equiv 16\frac{3 \beta ^2-2\alpha\chi}{\alpha}, \\
q&\equiv 4\frac{\alpha ^2 \mu + 2\alpha  \beta  \chi +2 \beta ^3}{\alpha ^3}, \\
Q&\equiv \Bigg[ \sqrt{
\left(27 \alpha  \mu ^2+27 \beta ^2 \nu -18 \alpha  \chi  \nu 
+36 \beta  \chi  \mu +8 \chi^3\right)^2
- \Xi^3} \notag\\
&\hspace{1cm}+\left(27 \alpha  \mu ^2+27 \beta ^2 \nu -18 \alpha  \chi  \nu 
+36 \beta  \chi  \mu +8 \chi^3\right)
\Bigg]^{1/3},
\end{align}
with
\begin{align}
\Xi&\equiv 3 \alpha  \nu +12 \beta  \mu +4\chi^2, \\
\chi&\equiv g_1^2\left(18 g_1^2-19 g_1 g_3 -6 (7 g_2^2+3 g_3^2)\right)+g_1(-20 g_2^2 g_3+13 g_3^3)\notag\\
&\quad+20 g_2^4-26 g_2^2 g_3^2+ 6 g_3^4, \\
\mu&\equiv g_3^2 \left(-2 g_3^2-3g_1 g_3+5 g_1^2+6 g_2^2\right)+2 g_1 g_3 \left(g_1^2+g_2^2\right)-2 \left(g_1^2-g_2^2\right)^2,\\
\nu&\equiv g_3^2 \left(g_3^2+g_1 g_3-2 \left(g_1^2+g_2^2\right)\right).
 \end{align}


\end{document}